\documentclass{icmart}

\usepackage{url}
\usepackage{verbatim}

\contact[alberto.cattaneo@math.unizh.ch]{Institut f\"ur Mathematik, Universit\"at Z\"urich, Winterthurerstrasse 190, CH-8057 Z\"urich, 
Switzerland}

\numberwithin{equation}{section}

\newtheorem{theorem}{Theorem}[section]

\newtheorem{lemma}[theorem]{Lemma}

\theoremstyle{definition}
\newtheorem{definition}[theorem]{Definition}
\newtheorem{remark}[theorem]{Remark}
\newtheorem{example}[theorem]{Example}
\newtheorem{notation}[theorem]{Notation}

\theoremstyle{remark}
\newtheorem*{Ack}{Acknowledgment}

\newcommand{\Hom}{\operatorname{Hom}}

\newcommand{\fine}{\hfill$\triangle$}

\newcommand{\End}{\operatorname{End}}
\newcommand{\Ber}{\operatorname{Ber}}
\newcommand{\Der}{\operatorname{Der}}
\newcommand{\D}{\operatorname{D}}
\newcommand{\Div}{\mathrm{div}}
\newcommand{\BER}{\operatorname{BER}}
\newcommand{\Spec}{\operatorname{Spec}}
\newcommand{\ad}{\operatorname{ad}}
\newcommand{\rank}{\operatorname{rank}}

\newcommand{\Mor}{\operatorname{Mor}}
\newcommand{\Map}{\operatorname{Map}}
\newcommand{\HKR}{\operatorname{HKR}}
\newcommand{\sign}{\operatorname{sign}}

\newcommand{\ndash}{\nobreakdash-\hspace{0pt}}
\newcommand{\Ndash}{\nobreakdash--}

\newcommand{\calL}{\mathcal{L}}
\newcommand{\calO}{\mathcal{O}}
\newcommand{\calM}{\mathcal{M}}
\newcommand{\calN}{\mathcal{N}}

\newcommand{\calV}{\mathcal{V}}
\newcommand{\calD}{\mathcal{D}}
\newcommand{\calX}{\mathcal{X}}

\newcommand{\frX}{\mathfrak{X}}
\newcommand{\frg}{{\mathfrak{g}}}
\newcommand{\frh}{{\mathfrak{h}}}

\newcommand{\sfB}{{\mathsf{B}}}

\newcommand{\sfS}{{\mathsf{S}}}

\newcommand{\Hoch}{{\mathsf{Hoch}}}

\newcommand{\sfeta}{\boldsymbol{\eta}}
\newcommand{\sfxi}{\boldsymbol{\xi}}

\newcommand{\bbR}{{\mathbb{R}}}

\newcommand{\bbZ}{{\mathbb{Z}}}
\newcommand{\bbN}{{\mathbb{N}}}

\newcommand{\ii}{{\mathrm{i}}}
\newcommand{\dd}{{\mathrm{d}}} 
\newcommand{\ee}{{\mathrm{e}}}
\newcommand{\LL}{{\mathrm{L}}}
\newcommand{\de}{{\partial}}

\newcommand{\bC}{{\mathbf C}}
\newcommand{\bfrX}{{\boldsymbol{\frX}}}
\newcommand{\bcalX}{{\boldsymbol{\calX}}}
\newcommand{\scalO}{{\boldsymbol{\calO}}}
\newcommand{\bcalD}{{\boldsymbol{\calD}}}
\newcommand{\bcalV}{{\boldsymbol{\calV}}}

\newcommand{\braket}[2]{\left\langle{\,{#1}\,,\,{#2}\,}\right\rangle}
\newcommand{\vev}[1]{{\left\langle\,{#1}\,\right\rangle}}
\newcommand{\Lie}[2]{{\left[{\,{#1}\,,\,{#2}\,}\right]}}
\newcommand{\BV}[2]{{\left({\,{#1}\,,\,{#2}\,}\right)}}
\newcommand{\Poiss}[2]{\left\{{\,{#1}\,,\,{#2}\,}\right\}}

\newcommand\qq{\em}
\newcommand\cmp[1]{{\qq Commun.\ Math.\ Phys.\ \bf #1}}

\newcommand\pl[1]{{\qq Phys.\ Lett.\ \bf #1}}

\newcommand\mpl[1]{{\qq Mod.\ Phys.\ Lett.\ \bf #1}}

\newcommand\lmp[1]{{\qq Lett.\ Math.\ Phys.\ \bf #1}}

\newcommand\cqg[1]{{\qq Class.\ Quant.\ Grav.\ \bf #1}}
\newcommand\prept[1]{{\qq Phys.\ Rept.\ \bf #1}}

\newcommand\anp[1]{{\qq Ann.\ Phys.\ \bf #1}}

\newcommand\dmj[1]{{\qq Duke Math.\ J. \bf #1}}

\newcommand\aihp[1]{{\qq Ann.\ Inst.\ Henri Poincar\'e \bf #1}}
\newcommand\inm[1]{{\qq  Invent.\ Math.\ \bf #1}}
\newcommand\conm[1]{{\qq  Contemp.\ Math.\ \bf #1}}

\title[TFT, Quantization, and Reduction]{From Topological Field Theory to Deformation Quantization and Reduction}

\author[A.~S.~Cattaneo]{Alberto S.~Cattaneo\thanks{The author acknowledges partial support of SNF Grant No.~200020-107444/1.}}

\begin{document}

\begin{abstract}
This note describes the functional-integral quantization of two-dimensional topological
field theories together with applications to problems in deformation quantization of Poisson manifolds
and reduction of certain submanifolds. A brief introduction to smooth graded manifolds and to the Batalin--Vilkovisky
formalism is included.
\end{abstract}



\maketitle

\section{Introduction: a 2D TFT}
\subsection{The basic setting}
Let $\Sigma$ be a smooth compact $2$\ndash manifold. On $\calM_1:=\Omega^0(\Sigma)\oplus\Omega^1(\Sigma)$ one may define
the following very simple action functional:
\begin{equation}\label{e:S1}
S(\xi,\eta):=\int_\Sigma\eta\,\dd\xi,\qquad
\xi\in\Omega^0(\Sigma),\ 
\eta\in\Omega^1(\Sigma),
\end{equation}
which is invariant under the distribution $\{0\oplus\dd\beta$, $\beta\in\Omega^0(\Sigma)\}$.
Denoting by $\delta_\beta$ the constant section $0\oplus\dd\beta$ and 
taking
$\xi$ and $\eta$ 
as coordinates on $\calM_1$,
we have
\begin{equation}\label{e:delta1}
\delta_\beta\xi=0,\quad\delta_\beta\eta=\dd\beta.
\end{equation}

The critical points are closed $0$\ndash\ and $1$\ndash forms. As symmetries are given by exact forms,
the space of solutions modulo symmetries, to which we will refer as the moduli space of solutions, 
is $H^0(\Sigma)\oplus H^1(\Sigma)$, which is finite dimensional.
Moreover, it depends only on the topological type of $\Sigma$. Actually, something more is true: the action
of the group of diffeomorphisms connected to the identity is included in the symmetries restricted to
the submanifold of critical points.
In fact, for every vector field $Y$ on $\Sigma$, we have
$\LL_Y\xi = \iota_Y\dd\xi$ and $\LL_Y\eta=\iota_Y\dd\eta+\dd\iota_Y\eta$.
So upon setting $\dd\xi=\dd\eta=0$, we get $\LL_Y=\delta_{\beta_Y}$ with $\beta_Y=\iota_Y\eta$.
This is the simplest example of $2$-dimensional topological field theory (TFT) that contains
derivatives in the fields.\footnote{This example belongs to the larger class of so-called $BF$ theories.
This is actually a $2$-dimensional abelian $BF$ theory.}

One may also allow $\Sigma$ to have a boundary $\de\Sigma$.
If we do not impose boundary conditions, the variational problem yields the extra condition $i)$ $\iota_{\de\Sigma}^*\eta=0$
where $\iota_{\de\Sigma}$ denotes the inclusion map of $\de\Sigma$ into $\Sigma$.
So it makes sense to impose $i)$ from the beginning. The second possibility is to impose the boundary condition $ii)$ that 
$\xi_{|_{\de\Sigma}}$ should be constant. By translating $\xi$, we may always assume this constant to be zero.\footnote{For simplicity, in this note we do not
consider the case \cite{CFb} when the boundary is divided into different components with different boundary conditions.}
For the symmetries to be consistent with boundary conditions $i)$,
we have to assume that $\beta_{|_{\de\Sigma}}$ is constant, and again we may assume without loss of generality that this constant vanishes.
So we consider the following two cases:
\begin{align*}
 &\text{Neumann boundary conditions:} & \iota_{\de\Sigma}^*\eta &=0, &  \beta_{|_{\de\Sigma}}&=0 & (N)\\
 &\text{Dirichlet boundary conditions:}  & \xi_{|_{\de\Sigma}} &=0, & &  & (D)
\end{align*}

\subsection{Generalizations}
To make things more interesting, we may replicate $n$ times what we have done above. Namely, take 
$\calM_n=\calM_1^n$ and define
\[
S(\{\xi\},\{\eta\}):=\int_\Sigma\sum_{I=1}^n\eta_I\,\dd\xi^I,\qquad
\xi^I\in\Omega^0(\Sigma),\ 
\eta_I\in\Omega^1(\Sigma).
\]
Identifying $\calM_n$ with $\Omega^0(\Sigma,\bbR^n)\oplus\Omega^1(\Sigma,(\bbR^n)^*)$,
we may also write 
\begin{equation}\label{e:SRn}
S(\xi,\eta):=\int_\Sigma \braket\eta{\dd\xi},\qquad
\xi\in\Omega^0(\Sigma,\bbR^n),\ 
\eta\in\Omega^1(\Sigma,(\bbR^n)^*),
\end{equation}
where $\braket{\ }{\ }$ denotes the canonical pairing.
The symmetries are now defined by the addition to $\eta$ of an exact $1$\ndash form $\dd\beta$,
$\beta\in\Omega^1(\Sigma,(\bbR^n)^*)$.
If $\Sigma$ has a boundary, we then choose N or D boundary conditions for each value of the index $I$. 
We may also modify the action functional by adding the local term 
\begin{equation}\label{e:Salpha}
S_\alpha(\xi,\eta)= \frac12\int_\Sigma \alpha(\xi)(\eta,\eta),
\end{equation}
where $\alpha$ is a smooth map $\bbR^n\to\Lambda^2\bbR^n$ or more generally an element of $\Hat S((\bbR^n)^*)\otimes \Lambda^2\bbR^n$,
where $\Hat S((\bbR^n)^*)$ denotes the formal completion (i.e., the space of formal power series) of the symmetric algebra $S((\bbR^n)^*)$.
We will discuss in the following under which assumptions on $\alpha$ and
on the boundary conditions this term may be added without breaking the symmetries of $S$.

A further generalization with a smooth $n$\ndash manifold $M$ as target exists. The space
$\calM(M):=\{\text{bundle maps }T\Sigma\to T^*M\}$ fibers over 
$\Map(\Sigma,M)$
with fiber at a map $X$ the space of sections $\Gamma(T^*\Sigma\otimes X^*T^*M)$.
Regarding $\dd X$ as a section of $T^*\Sigma\otimes X^*TM$ and using the canonical pairing $\braket{\ }{\ }$ of $TM$ with $T^*M$, we define
\begin{equation}\label{e:SM}
S(X,\eta):=\int_\Sigma \braket\eta{\dd X},\qquad
X\in\Map(\Sigma,M),\ 
\eta\in\Gamma(T^*\Sigma\otimes X^*T^*M).
\end{equation}
The critical points are now given by pairs of a constant map $X$ and a closed form 
$\eta$ with $x=X(\Sigma)$. The symmetries are given by translating
$\eta$ by $\dd\beta$ with $\beta\in\Gamma(X^*T^*M)$.\footnote{The derivative of $\beta$ is computed
by choosing any torsion-free connection on $M$.} For the boundary conditions, one chooses
a submanifold $C$ of $M$ and imposes
$X(\de\Sigma)\subset C$ and $\iota_{\de\Sigma}^*\eta\in\Gamma(T^*\de\Sigma\otimes X^*N^*C)$,
where the conormal bundle $N^*C$ is by definition the annihilator of $TC$ as a subbundle of $T_CM$; viz.:
\begin{equation}\label{e:con}
N^*_xC:=\{\alpha\in T^*_xM : \alpha(v)=0\ \forall v\in T_xC\},\ x\in C.
\end{equation}
Accordingly, we require $\iota_{\de\Sigma}^*\beta\in\Gamma(X^*N^*C)$.
Observe that the tangent space at a given solution (i.e., $X(\Sigma)=x$, $\eta$ closed),
is isomorphic---upon choosing local coordinates around $x$---to $\calM_n$, just by setting $X=x+\xi$.
Moreover,
the action evaluated around a solution is precisely \eqref{e:SRn}.

A global generalization of \eqref{e:Salpha} is also possible. Namely, to every bivector field $\pi$
(i.e., a section of $\Lambda^2 TM$), we associate the term
\begin{equation}\label{e:Spi}
S_\pi(X,\eta) = \frac12\int_\Sigma \pi(X)(\eta,\eta).
\end{equation}
If we work in the neighborhood of a solution $x$ and set $X=x+\xi$, then \eqref{e:Spi} reduces to
\eqref{e:Salpha} with $\alpha(v)=\pi(x+v)$, $\xi\in \bbR^n\simeq T_xM$. Actually we are interested in working in a formal
neighborhood, so we set $\alpha$ to be the Taylor expansion of $\pi$ around $x$ and regard it as an element of
$\Hat S(\bbR^n)^*\otimes \Lambda^2\bbR^n$.

\subsection{Functional-integral quantization}
The action functional \eqref{e:SM} is not very interesting classically. 
Much more interesting is its quantization, 
by which we mean the evaluation of  ``expectation values'', i.e., ratios of functional integrals 
\begin{equation}\label{e:vev}
\vev\calO_\text{cl}:=\frac{\int_{\calM(M)} \ee^{\frac\ii\hbar S}\, \calO}{\int_{\calM(M)} \ee^{\frac\ii\hbar S}},
\end{equation}
where $\calO$ is a function (which we assume to be a polynomial or a formal power series) on $\calM(M)$. 
The evaluation of these functional integrals
consists of an ordinary integration over the moduli space of solutions and of an ``infinite-dimensional integral''
which is operatively defined in terms of the momenta of the Gaussian distribution given by $S$.

The finite-dimensional integration 
is not problematic, though it requires choosing a measure on the moduli space of solution. The main assumption in this paper
is that the first cohomology of $\Sigma$ with whatsoever boundary conditions is trivial. 
Actually, we assume throughout that $\Sigma$ is the $2$\ndash disk $D$. So up to equivalence a solution is given by specifying
the value $x$ of the constant map $X$, and the moduli space of solutions is $M$. We then choose a delta measure on $M$ at some point $x$.

The second integration, performed around a point $x$, is then over $\calM_n$.
The main problem is that the operator $\dd$ defining the quadratic form in $S$ is not invertible. To overcome this problem
and make sense of the integration, we resort to the so-called BV (Batalin--Vilkovisky \cite{BV})
formalism, which is reviewed in Sect.~\ref{s:BV}.
Besides giving us an operative unambiguous definition of \eqref{e:vev}, the BV formalism will also provide us
with relations among the expectation values, the so-called Ward identities (see Remark~\ref{r:Ward} 
and subsection~\ref{s:Ward}). The latter computation is however less rigorous; one might
think of this as a machinery suggesting relations that have next to be proven to hold.
Moreover, the BV formalism leads naturally to the generalization when the target $M$ is a graded manifold (see Sect.~\ref{s:sgm}).
In this context there is an interesting duality (see~\ref{s:dual} and~\ref{s:hdt}) between different targets.

\begin{Ack}
I thank F.~Bonechi, D.~Fiorenza, F.~Helein, R.~Mehta, C.~Rossi, F.~Sch\"atz, J.~Stasheff and M.~Zambon for very useful comments.
\end{Ack}

\section{Smooth graded manifolds}\label{s:sgm}
In this Section we give a crash course in the theory of smooth graded manifolds. 
A graded manifold is a supermanifold with a $\bbZ$\ndash refinement of the $\bbZ_2$\ndash grading.
As we work in the smooth setting, we can work with algebras of global functions and so avoid the more technical definitions
in terms of ringed spaces. We begin with recalling some basic definitions and notations.

{\small
\subsection{Graded linear algebra}
A graded vector space $V$ is a direct sum over $\bbZ$ of vector spaces: $V=\oplus_{i\in\bbZ} V_i$.
Elements of $V_i$ have by definition degree $i$. 
By $V[n]$, $n\in\bbZ$, we denote the graded vector space with the same components of $V$ but shifted by $n$; i.e., 
$V[n]_i:=V_{i+n}$. 
A morphism $\phi\colon V\to W$ of graded vector spaces is a homomorphism that preserves degree: i.e.,
$\phi(V_i)\subset W_i$ $\forall i$. A $j$\ndash graded homomorphism $\phi\colon V\to W$ is a morphism $V\to W[j]$;
i.e., $\phi(V_i)\subset W_{i+j}$. We denote by $\Hom_j(V,W)$ the space of $j$\ndash graded homomorphisms. We may regard the vector space
of homomorphisms as a graded vector space 
$\Hom(V,W)=\oplus_j\Hom_j(V,W)$.
In particular, by regarding the ground field as a graded vector space concentrated in degree zero, 
the dual $V^*$ of a graded vector space $V$ is also naturally graded with $V^*_i:=(V^*)_i$ 
isomorphic to $(V_{-i})^*$. Observe that $V[n]^*=V^*[-n]$.
Tensor products of graded vector spaces are also naturally graded: 
$(V\otimes W)_i=\oplus_{r+s=i} V_r\otimes W_s$.

\subsubsection{Graded algebras}
A graded algebra $A$ is an algebra which is also a graded vector space such that the product is a morphism of graded vector spaces.
The algebra is called graded commutative (skew-commutative)
if $ab=(-1)^{ij}ba$ ($ab=-(-1)^{ij}ba$)
for all $a\in A_i$, $b\in A_j$, $i,j\in\bbZ$.
The symmetric algebra of a graded vector space is the graded commutative algebra defined as
$S(V)=T(V)/I$, where $T(V)$ denotes the tensor algebra and $I$ is the two-sided ideal generated by
$vw-(-1)^{ij}wv$, $v\in V_i$, $w\in V_j$.
We denote by $\Hat S(V)$ its formal completion consisting of formal power series.

A graded skew-commutative algebra is
called a graded Lie algebra (GLA)
if its product, denoted by $\Lie{\ }{\ }$ satisfies the graded Jacobi identity:
$\Lie a{\Lie bc}=\Lie{\Lie ab}c + (-1)^{ij}\Lie b{\Lie ac}$, 
for all $a\in A_i$, $b\in A_j$, $c\in A$, $i,j\in\bbZ$.

\subsubsection{Graded modules}
A graded module $M$ over a graded algebra $A$ is a graded vector space which is a
module over $A$ regarded as a ring
such that the action $A\otimes M\to M$ is a morphism of graded vector spaces. If $M$ is a module, then so is $M[j]$
for all $j\in\bbZ$.

The tensor product $M_1\otimes_A M_2$ over $A$
of a right $A$\ndash module $M_1$ and a left $A$\ndash module $M_2$ is defined as the quotient of $M_1\otimes M_2$ by
the submodule generated by $m_1a\otimes m_2-m_1\otimes am_2$, for all $a_\in A$, $m_i\in M_i$. If $M_1$ and $M_2$ are bimodules,
then so is $M_1\otimes_A M_2$.

Let $M$ be a left $A$\ndash module.
If $A$ is graded commutative (skew-commutative), we make $M$
into a bimodule by setting $ma:=(-1)^{ij}am$ ($ma:=-(-1)^{ij}am$) , $a\in A_i$, $m\in M_j$.
We may regard $A\oplus M$ as a graded commutative (skew-commutative) algebra
by setting the product of two elements in $M$ to zero. If $A$ is a GLA, then so is $A\oplus M$. 

Let $A$ be graded commutative.
For every $A$\ndash module $M$, we define inductively the $A$\ndash module $T^k_A(M)$ as $T^{k-1}_A(M)\otimes_A M$, 
with $T^0_A(M):=A$.
So one gets the graded associative algebra $T_A(M):=\oplus_{j\in\bbN}T^j_A(M)$ which is also an $A$\ndash bimodule. 
The symmetric algebra $S_A(M)$ is defined
as the quotient of $T_A(M)$ by the two-sided ideal generated by
$vw-(-1)^{ij}wv$, $v\in M_i$, $w\in M_j$.
We denote by $\Hat S_A(M)$ its formal completion.

\subsubsection{Derivations and multiderivations}
A $j$\ndash graded endomorphism $D$ of a grad\-ed algebra $A$ is called a $j$\ndash graded derivation if
$D(ab)=D(a)b+(-1)^{ij}aD(b)$ for all $a\in A_i$, $i\in\bbZ$, and all $b\in A$. For example, if $A$ is a GLA,
$\Lie a{\ }$ is an $i$\ndash graded derivation for every $a\in A_i$. A differential is
a derivation of degree $1$ that squares to zero. A differential graded Lie algebra (DGLA) is a GLA with a differential. 

We denote by $\Der_j(A)$ the space of $j$\ndash graded
derivations of a graded algebra $A$ and set $\Der(A)=\oplus_{j\in\bbZ}\Der_j(A)$. It is a GLA with bracket
$\Lie{D_1}{D_2}:=D_1D_2-(-1)^{j_1j_2}D_2D_1$, $D_i\in\Der_{j_i}(A)$. Observe that $\Der(A)$ is a left $A$\ndash module
while $A$ is a left $\Der(A)$\ndash module. Thus, for every $n$,
we may regard $\Der(A)\oplus A[n]$ as a GLA
with the property
\begin{equation}\label{e:Xfg}
\Lie X{fg}=(-1)^{jk}f\Lie Xg + \Lie Xf g,\qquad
\forall X\in\Der(A)_j,\ f\in A_k,\ g\in A.
\end{equation}

Given a graded commutative algebra $A$,
we define the algebra $\Hat\D(A,n)$ of $n$\ndash shifted multiderivations
by $\Hat\D(A,n):=\Hat S_A(\Der(A)[-n])$,
and denote by $\D(A,n)$ its subalgebra $S_A(\Der(A)[-n])$.
Observe that the GLA structure on $\Der(A)\oplus A[n]$ can be extended to $\D(A,n)[n]$ 
and to $\Hat\D(A,n)[n]$ 
in a unique way, compatible with \eqref{e:Xfg}, such that
\[
\Lie{D_1}{{D_2}{D_3}}=(-1)^{(j_1+n)j_2}D_2\Lie{D_1}{D_3} +\Lie{D_1}{D_2}D_3,
\quad D_i\in\D(A)_{j_i}.
\]
By this property, $\Hat\D(A,n)$ is a so-called $n$-Poisson algebra. For $n=0$, it is a graded Poisson algebra.
A $1$-Poisson algebra is also called a Gerstenhaber algebra. Since 
this case is particularly important, we will use the special notation $\Hat\D(A)$ ($\D(A)$) for
$\Hat\D(A,1)$ ($\D(A,1)$). Elements of $\Hat\D(A)$ are simply called multiderivations. More precisely, elements
of $S_A^j(\Der(A)[-1])$ are called $j$\ndash derivations, and a $j$\ndash derivation is said to be of degree $k$
and of total degree $j+k$ if it belongs to $\Hat\D(A)_{j+k}$.
More generally, elements of $S_A^j(\Der(A)[-n])_{k+nj}$ are called $n$\ndash shifted $j$\ndash derivations of degree $k$.

Given an $n$\ndash Poisson algebra $(P,\bullet,\Lie{\ }{\ })$, one defines $\ad\colon P\to\Der(P)$ by
$\ad_XY:=\Lie XY$, $X,Y\in P$. The $n$\ndash Poisson algebra is said to be nondegenerate if $\ad$ is surjective
(in other words, if the first Lie algebra cohomology of $P$ with coefficients in its adjoint representation
is trivial).

\subsubsection{The Hochschild complex}
Given a graded vector space $A$ define $\Hoch^{j,m}(A)=\Hom_j(A^{\otimes m},A)$,
$\Hoch^n(A)=\bigoplus_{j+m=n}\Hoch^{j,m}(A)$, and the Hochschild complex $\Hoch(A)=\bigoplus_n \Hoch^n(A)$.
One may compose elements of $\Hoch(A)$ as follows: given $\phi\in \Hoch^{j_1,m_1}$ and $\psi\in \Hoch^{j_2,m_2}$,
one defines the nonassociative product 
\[
\phi\bullet\psi=(-1)^{(j_2+m_2-1)(m_1-1)}\sum_i(-1)^{i(m_2-1)}  \phi\circ(1^{\otimes i}\otimes\psi\otimes1^{\otimes(m_1-1-i)})\in \Hoch^{j_1+j_2,m_1+m_2-1}.
\]
It turns out that
its associated bracket $\Lie\phi\psi:=\phi\bullet\psi-(-1)^{(j_1+m_1-1)(j_2+m_2-1)}\psi\bullet\phi$ 
makes $\Hoch(A)[1]$ into a GLA\@. A product on $A$ is an element $\mu$ of $\Hoch^{0,2}(A)$. Define $b=\Lie\mu{\ }$.
Then $b$ is a differential on $\Hoch(A)[1]$ if{f} the product is associative.

\subsubsection{Differential and multidifferential operators}\label{s:mdo}
Given a graded associative algebra $A$ and graded derivations $\phi_i\in\Der(A)_{j_i}$, the composition $\phi_1\circ\dots\circ\phi_k$
is an element of $\Hoch^{j_1+\dots+j_k,1}$. 
A differential operator on $A$ is by definition
a linear combination of homomorphisms of this form. A multidifferential operator
is a linear combination of elements of $\Hoch(A)$ of the form 
$(a_1,\dots,a_n)\mapsto \phi_1(a_1)\dots\phi_n(a_n)$ where each $\phi_i$ is a differential operator. Denote by $\calD(A)$
the Lie subalgebra of multidifferential operators in $\Hoch(A)[1]$. As the product is a multidifferential operator itself,
$\calD(A)$ is also a subcomplex of $(\Hoch(A)[1],b)$. 
For $A$ graded commutative, one defines the HKR map (Hochschild--Kostant--Rosenberg \cite{HKR})
$\HKR\colon\D(A)\to\calD(A)$ as the linear extension of 
\[
\phi_1\cdots\phi_n\mapsto\left(a_1\otimes\dots\otimes a_n\mapsto\sum_{\sigma\in S_n}\sign(\sigma)\phi_{\sigma(1)}(a_1)\cdots\phi_{\sigma(n)}(a_n)\right),
\]
where the $\phi_i$s are derivations and the sign is defined by $\phi_{\sigma(1)}\cdots\phi_{\sigma(n)}=\sign(\sigma)\phi_1\cdots\phi_n$ in $\D(A)$.
It turns out that $\HKR$ is a chain map $(\D(A),0)\to(\calD(A),b)$.
It is a classical result \cite{HKR} that in certain cases (e.g., when $A$ is the algebra of smooth functions on a smooth manifold), $\HKR$ is a quasiisomorphism
(i.e., it induces an isomorphism in cohomology).
}

\subsection{Graded vector spaces}
{}From now on we assume the ground field to be $\bbR$. 
For simplicity we consider only finite-dimensional vector spaces. 
We define the algebra of polynomial functions over a graded vector space $V$ as the symmetric algebra of $V^*$ and
the algebra of smooth functions as its formal completion. We use the notations
$\bC^\infty(V):=S(V^*)\subseteq \Hat\bC^\infty(V):=\Hat S(V^*)$.
Elements of $S^0(V^*)\simeq\bbR$ will be called constant functions. 

\subsubsection{Multivector fields}
A vector field on $V$ is by definition a linear combination of graded derivations on its algebra of functions.
We use the notations $\bfrX(V):=\Der(\bC^\infty(V))$,
$\Hat\bfrX(V):=\Der(\Hat\bC^\infty(V))$.
Observe that we may identify $\bfrX(V)$ and $\Hat\bfrX(V)$ with $\bC^\infty(V)\otimes V$
and $\Hat\bC^\infty(V)\otimes V$, respectively. Elements of $S^0(V^*)\otimes V\simeq V$ will be called constant vector fields.

Multivector fields are by definition multiderivations. In particular, $k$\ndash vector fields are $k$\ndash derivations,
and we define their degree and total degree correspondingly.
We use the notations
$\bcalX(V):=\D(\bC^\infty(V))$ and
$\Hat\bcalX(V):=\Hat\D(\Hat\bC^\infty(V))$
for the corresponding Gerstenhaber algebras. We also define the $n$\ndash Poisson algebras $\bcalX(V,n)$ and 
$\Hat\bcalX(V,n)$ of $n$\ndash shifted multivector fields
as $\D(\bC^\infty(V),n)$ and
$\Hat\D(\Hat\bC^\infty(V),n)$.
We have the following identifications:
\begin{subequations}\label{e:calXV}
\begin{alignat}3
\bcalX(V,n)&\simeq S(V^*)&\otimes S(V[-n])&\simeq \bC^\infty(V\oplus V^*[n]),\\
\Hat\bcalX(V,n)&\simeq \Hat S(V^*)&\Hat\otimes
\Hat S(V[-n])&\simeq \Hat\bC^\infty(V\oplus V^*[n]).
\end{alignat}
\end{subequations}

\subsubsection{Berezinian integration}\label{s:ber}
Let $V$ be an odd vector space (i.e., a graded vector space with
nontrivial components only in odd degrees). By integration we simply mean a linear form on its
space of functions $\bC^\infty(V)=\Hat\bC^\infty(V)$, which is isomorphic, forgetting degrees, to 
$\Lambda  V^*$.\footnote{By $\Lambda V$, 
we mean the usual exterior algebra of $V$ regarded as an ordinary vector space, i.e., forgetting degrees.}
So integration is defined by an element $\mu$ of $\Lambda V$.
We use the notation $\int_V f\,\mu$ for the pairing $\braket f\mu$.
We call an element of $\Lambda V$ a Berezinian form if its component in $\Lambda^\mathrm{top} V$, $\mathrm{top}=\dim V$, 
does not vanish. In this case integration has the property that its restriction to the space of functions of top degree is injective.
A Berezinian form concentrated in top degree, i.e., an element of $\Lambda^\mathrm{top}V\setminus\{0\}$, is  called pure and
has the additional property that the corresponding integral vanishes on functions that are not of top degree.
Observe that a pure Berezinian form $\rho$ establishes an isomorphism
$\phi_\rho\colon\bC^\infty(V)\simeq\Lambda V^*\stackrel\sim\to\Lambda V$,
$g\mapsto\iota_g\rho$.
If $\mu=\iota_g\rho$, then
$\int_V f\,\mu =\braket f{\iota_g\rho}=\int_V fg\,\rho$,
so we simply write $g\rho$ instead of $\iota_g\rho$.
\begin{lemma}\label{l:ber}
Given a pure Berezinian form $\rho$, for every Berezinian form $\mu$ there is a unique constant $c\not=0$ and
a unique function $\sigma\in\Lambda^{>0}V^*$ such that $\mu=c\,\ee^\sigma\rho$.
\end{lemma}
\begin{proof}
Set $g=\phi_\rho^{-1}(\mu)$.
If $\mu$ is a Berezinian form,  its component $c$ in $\Lambda^0V^*$ is invertible.
So we may write,
$g=c(1+h)$ with $h\in\Lambda^{>0}V^*$. Finally we define $\sigma=\log(1+h)=\sum_{k=1}^\infty (-1)^{k+1}h^k/k$ (observe that this is
actually a finite sum).
\end{proof}
\begin{lemma}\label{l:div}
For every Berezinian form $\mu$, there is a map
$\Div_\mu\colon\bfrX(V)\to\bC^\infty(V)$ (the divergence operator) such that
\[
\int_V X(f)\,\mu=-\int_V f\,\Div_\mu X\,\mu,\qquad
\forall f\in\bC^\infty(V).
\]
Moreover, $\Div_{c\mu}=\Div_\mu$ for every constant $c\not=0$. In particular, all pure Berezinian forms define
the same divergence operator.
\end{lemma}
\begin{proof}
The map $f\mapsto\int_V X(f)\,\mu$ is linear. So there is a unique $\mu_X\in\Lambda V$ such that
$\int_V X(f)\,\mu=\int_V f\,\mu_X$. Given a pure Berezinian form $\rho$, define $g_\mu=\phi_\rho^{-1}(\mu)$ and
$g^\mu_X=\phi_\rho^{-1}(\mu_X)$. Thus, $\mu_X=g^\mu_X\rho=g^\mu_X g_\mu^{-1}\mu$. Then we define
$\Div_\mu X$ as $-g^\mu_X g_\mu^{-1}$. Observe that this does not depend on the choice of $\rho$.
\end{proof}
\subsection{Graded vector bundles}
A graded vector bundle is a vector bundle whose fibers are graded vector spaces and such that
the transition functions are morphisms of graded vector spaces. All the constructions for graded vector
spaces described above extend to graded vector bundles. In particular, given a graded vector bundle
$E$, we may define the shifted graded vector bundles $E[n]$, the dual bundle $E^*$ (and $E[n]^*=E^*[-n]$),
the symmetric algebra bundle $S(E)$ and its formal completion $\Hat S(E)$. We also define
the graded commutative algebras of functions (we restrict for simplicity to graded vector bundles of finite rank) 
accordingly in terms of sections
$\bC^\infty(E):=\Gamma(S(E^*))\subseteq \Hat\bC^\infty(E):=\Gamma(\Hat S(E^*))$.
Elements of $\bC^\infty(E)$ will be called polynomial functions.

\begin{remark}\label{r:Tn}
In case the given vector bundle is the tangent or the cotangent bundle of a manifold $M$, it is customary
to write the shift after the $T$ symbol; viz., one writes $T[n]M$ and $T^*[n]M$ instead of $TM[n]$ and $T^*M[n]$.
We have 
$\bC^\infty(T[1]M)=\Hat\bC^\infty(T[1]M)=\Omega(M)$ and
$\bC^\infty(T^*[1]M)=\Hat\bC^\infty(T^*[1]M)=\calX(M)$,
where $\Omega(M)=\Gamma(\Lambda T^*M)$ and $\calX(M)=\Gamma(\Lambda TM)$ 
denote the graded commutative algebras of differential forms and of multivector
fields respectively. Observe that, in terms of graded vector spaces, we have
\begin{equation}\label{e:OmegaM}
\Omega(M) = \bigoplus_{i=0}^{\dim M} \Omega^i(M)[-i],\qquad
\calX(M) = \bigoplus_{i=0}^{\dim M} \calX^i(M)[-i],
\end{equation}
where $\Omega^i(M)$ and $\calX^i(M)$ are regarded as ordinary vector spaces (i.e., concentrated in degree zero).
\fine\end{remark}

\subsubsection{Multivector fields}\label{s:mvf}
A vector field on $E$ is a linear combination of graded derivations on its algebra of functions.
We use the notations
$\bfrX(E):=\Der(\bC^\infty(E))$,
$\Hat\bfrX(E):=\Der(\Hat\bC^\infty(E))$.
A vector field $X$ on $E$ is completely determined by its restrictions $X_M$
to $C^\infty(M)$ and $X_E$ to $\Gamma(E^*)$. Observe that $X_M$ is
a $\Hat\bC^\infty(E)$\ndash valued
vector field on $M$. Picking a connection $\nabla$ on $E^*$, we set
$X^\nabla_E(\sigma):=X(\sigma)-\nabla_{X_M}\sigma$, $\forall\sigma\in\Gamma(E^*)$.
Since $X^\nabla_E$ is $C^\infty(M)$\ndash linear, it defines
a bundle map $E^*\to\Hat S(E^*)$. The map $X\mapsto X_M\oplus X^\nabla_E$ is then
an isomorphism from $\Hat\bfrX(E)$ to $\Gamma(\Hat SE^*\otimes(TM\oplus E))$.
\begin{remark}\label{r:vert}
We may extend $\nabla$ to the whole of $\Hat\bC^\infty(E)$ as a derivation. So $\nabla_{X_M}$, unlike $X_M$, is a vector field on
$E$. The difference $X^\nabla:=X-\nabla_{X_M}$, which we call the vertical component of $X$,
is then also a vector field with the additional property that its restriction
to $C^\infty(M)$ vanishes. 
\fine\end{remark}
Multivector fields are by definition multiderivations. In particular, $k$\ndash vector fields are $k$\ndash derivations,
and we define their degree and total degree correspondingly.
We denote the corresponding Gerstenhaber algebras by
$\bcalX(E):=\D(\bC^\infty(E))$,
$\Hat\bcalX(E):=\Hat\D(\Hat\bC^\infty(E))$.
More generally, we define the $n$\ndash Poisson algebra $\Hat\bcalX(E,n)$ ($\bcalX(E,n)$) 
of $n$\ndash shifted (polynomial) multivector fields
as
$\Hat\D(\Hat\bC^\infty(E),n)$ ($\D(\bC^\infty(E),n)$).
Upon choosing a connection $\nabla$, we have the identifications
\begin{gather*}
\bcalX(E,n)\simeq \Gamma(SE^*)\otimes \Gamma(S((TM\oplus E)[-n]))\simeq \bC^\infty(E\oplus T^*[n]M\oplus E^*[n]),\\
\Hat\bcalX(E,n)\simeq \Gamma(\Hat S(E^*))\Hat\otimes \Gamma(\Hat S(TM\oplus E)[-n])\simeq \Hat\bC^\infty(E\oplus T^*[n]M\oplus E^*[n]).
\end{gather*}


\subsubsection{The Berezinian bundle}
We may easily extend the Berezinian integration introduced in~\ref{s:ber}
to every odd vector bundle $E\to M$ (i.e., a  bundle of odd vector spaces).
A section $\mu$ of the ``Berezinian bundle" $\BER(E):=\Lambda E\otimes\Lambda^\mathrm{top}T^*M$, $\mathrm{top}=\dim M$,
defines\footnote{We consider
$M$ to be orientable, otherwise replace the space of top forms with the space of densities.}
a $C^\infty(M)$\ndash linear map $\braket{\ }\mu\colon\bC^\infty(E)\simeq\Gamma(\Lambda E^*)\to\Omega^\mathrm{top}(M)$.
We set
$\int_E f\,\mu := \int_M\braket f\mu$.
(For $M$ non compact, this of course makes sense only for certain functions.) 
%
Like in the case of odd vector spaces, we are 
interested in integrations that are nondegenerate on the subspace of functions
of top degree. These are determined by sections of the Berezinian bundle whose top component is nowhere vanishing.
We call such sections Berezinian forms.
A pure Berezinian form $\rho$ is then by definition a Berezinian form concentrated in top degree, i.e., a nowhere vanishing section of the
``pure Berezinian bundle" $\Ber(E):=\Lambda^\mathrm{top}E\otimes\Lambda^\mathrm{top}T^*M$ (with the first ``top" the rank of $E$). 
\begin{example}\label{E:T*odd}
Let $E=T^*[k]M$, with $k$ odd and with $M$ orientable and connected. Then $\Ber(E)=(\Lambda^\mathrm{top}T^*M)^{\otimes2}$. So there is a two-to-one correspondence
between volume forms on $M$ and pure Berezinian forms on $E$. Let $v$ be a volume form and $\rho_v$ the corresponding Berezinian form.
If we identify functions on $T^*[k]M$ with multivector fields, we may then compute
$\int_{T^*[k]M} X\,\rho_v = \int_M \phi_v(X)\,v$, with
$\phi_v\colon\calX(M)\stackrel\sim\to\Omega(M)$, $X\mapsto\iota_Xv$.
As a further example, consider the graded vector bundle $L_C:=N^*[k]C$, $k$ odd, where $C$ is a submanifold of $M$ and
$N^*C$ its conormal bundle (defined in \eqref{e:con}). Now $\Ber L_C\simeq\Lambda^\mathrm{top}N^*C\otimes\Lambda^\mathrm{top}T^*C\simeq\Lambda^\mathrm{top}T^*_CM$,
where $T^*_CM$ is the restriction of $T^*M$ to $C$. 
Thus, a volume form $v$ on $M$ also determines by restriction a pure Berezinian form on $L_C$
which we denote by $\surd\rho_v$  
as the correspondence is now linear instead of quadratic. We may identify functions on $L_C$ with sections of the exterior algebra of $NC$.
We then have $\int_{L_C} X\,\surd\rho_v= \int_C \phi_v(\Tilde X)$,
where $\Tilde X$ is any multivector field on $M$ extending a representative of $X$ in $\Gamma(\Lambda T_CM)$.
Finally, we have a canonically defined surjective
morphism
$\iota^*_{L_C}\colon\bC^\infty(T^*[k]M)\to\bC^\infty(L_C)$
obtained by restricting a multivector field to $C$ and modding out
its tangent components. One should think of $L_C$ as a submanifold (actually, a Lagrangian submanifold) of $T^*[k]M$ with inclusion map denoted by $\iota_{L_C}$. 
We then have
\begin{equation}\label{e:intiotaC}
\int_{L_c} \iota^*_{L_C}(X)\,\surd\rho_v=\int_C \phi_v(X),\qquad
\forall X\in\Gamma(\Lambda TM)\simeq\bC^\infty(T^*[k]M),
\end{equation}
with the r.h.s.\ defined to be zero if form degree and dimension do not match.
\fine\end{example}
A pure Berezinian form $\rho$ establishes an isomorphism
$\phi_\rho\colon\bC^\infty(E)\simeq\Gamma(\Lambda E^*)\stackrel\sim\to\Gamma(\BER(E))$,
$g\mapsto\iota_g\rho$.
If $\mu=\iota_g\rho$, then
$\int_E f\,\mu =\int_M\braket f{\iota_g\rho}=\int_E fg\,\rho$,
so we simply write $g\rho$ instead of $\iota_g\rho$. 
Lemmata~\ref{l:ber} and~\ref{l:div} generalize as follows:
\begin{lemma}\label{l:murho}
Given a pure Berezinian form $\rho$, for every Berezinian form $\mu$ there is a unique nowhere vanishing function $f\in C^\infty(M)$
and a unique function $\sigma\in\Gamma(\Lambda^{>0}E^*)$ such that $\mu=f\ee^\sigma\rho$. If $M$ is connected,
there is a unique function $\sigma\in\bC^\infty(E)$ such that $\mu=\ee^\sigma\rho$ or $\mu=-\ee^\sigma\rho$.
\end{lemma}
\begin{lemma}\label{l:mudiv}
Let $E\to M$ be an odd vector bundle with $M$ compact and orientable. Then,
for every Berezinian form $\mu$, there is a map
$\Div_\mu\colon\bfrX(E)\to\bC^\infty(E)$ (the divergence operator) such that
\[
\int_E X(f)\,\mu=-\int_E f\,\Div_\mu X\,\mu,\qquad
\forall f\in\bC^\infty(E).
\]
Moreover, $\Div_{c\mu}=\Div_\mu$ for every constant $c\not=0$. 
\end{lemma}
The proof of Lemma~\ref{l:murho} is exactly the same as the proof of Lemma~\ref{l:ber}.
The proof of Lemma~\ref{l:mudiv} goes as the proof of Lemma~\ref{l:div} if we may assume that
the map $f\mapsto \braket{X(f)}\mu$ is $C^\infty(M)$\ndash linear. This is the case only for a vertical vector field.
By using Remark~\ref{r:vert}, we write $X$ as $\nabla_{X_M}+X^\nabla$, and $X^\nabla$ is vertical.
By further writing $X_M$ as $\sum_i h_i X_M^i$,
with $h_i\in\bC^\infty(E)$ and $X_M^i\in\frX(M)$, and manipulating the integral carefully, we end up with terms which
are $C^\infty(M)$\ndash linear 
plus terms where we may apply the usual divergence theorem on $M$. 
The expression for $\Div_\mu X$ is then easily seen not to depend on the choices
involved in this argument.
\begin{remark}\label{r:divfX}
One may easily see that for every vector field $X$ and  every function $g$, 
the divergence of $gX$ is the sum (with signs) of $g\Div_\mu X$ and $X(g)$.
\fine\end{remark}
Integration over an arbitrary graded vector bundle is defined by splitting it into its odd part (where Berezinian integration may be defined)
and its even part (where the usual integration theory makes sense). 

\subsection{Smooth graded manifolds}
We are now ready to define smooth graded manifolds. 
We call a graded commutative algebra a graded algebra of smooth (polynomial) functions if it is isomorphic to
the algebra of (polynomial) functions of a graded vector bundle. 
Next we denote by $\widehat{\mathsf{GrSmFun}}$  ($\mathsf{GrSmFun}$) the category
whose objects are graded algebras of smooth (polynomial) functions and whose
morphisms are graded algebra morphisms.
Finally, we define the category $\widehat{\mathsf{SmoothGr}}$ ($\mathsf{SmoothGr}$) of smooth graded manifolds as the dual of $\widehat{\mathsf{GrSmFun}}$
($\mathsf{GrSmFun}$).
In particular, graded vector spaces and graded vector bundles may be regarded as smooth graded manifolds, i.e.,
as objects in $\widehat{\mathsf{SmoothGr}}$ or $\mathsf{SmoothGr}$ depending on which algebra of functions we associate to them.
\begin{notation}
If $A$ is an object of $\mathsf{GrSmFun}$, we write $\Spec(A)$ for the same object in $\mathsf{SmoothGr}$.
Vice versa, if we start with an object $\calM$ of $\mathsf{SmoothGr}$, we denote by
$\bC^\infty(\calM)$ the same object in $\mathsf{GrSmFun}$. We use the notations $\widehat\Spec$ and $\Hat\bC^\infty$ for the hatted categories.
We denote 
by $\widehat\Mor(\calM,\calN)$
($\Mor(\calM,\calN)$) the space of morphisms from $\calM$ to $\calN$ in $\widehat{\mathsf{SmoothGr}}$ ($\mathsf{SmoothGr}$).
\end{notation}
\begin{remark}\label{r:mor}
The spaces of morphisms $\widehat\Mor(\calM,\calN)$
($\Mor(\calM,\calN)$) may actually be given the structure of (possibly infinite-dimensional) smooth manifolds.
In particular, for $\calN=V$ a graded vector space, they may be regarded as (possibly infinite-dimensional) vector spaces:
\begin{equation}\label{e:mor}
\Mor(\calM,V) \simeq (V\otimes\bC^\infty(\calM))_0,\qquad
\widehat\Mor(\calM,V) \simeq (V\otimes\Hat\bC^\infty(\calM))_0,
\end{equation}
for $\bC^\infty(V)$ is generated by $V^*$, so
an algebra morphism from $\bC^\infty(V)$ is determined 
by its restriction to $V^*$ as a morphism 
of graded vector spaces.
\fine\end{remark}

By our definition, every smooth graded manifold may actually  be realized as a graded vector bundle though not in a canonical way.
One often obtains new graded algebras of smooth functions by some canonical constructions, yet their realization as
algebras of functions of graded vector bundles involves some choice.
\begin{example}\label{ex:Tstar}
As we have seen at the end of~\ref{s:mvf}, upon choosing a connection, we may identify
the algebra $\Hat\bcalX(E,n)$ of shifted multivector fields on $E$ with the graded algebra of smooth functions
on $E\oplus T^*[n]M\oplus E^*[n]$. 
We write $T^*[n]E$ for $\Spec\Hat\bcalX(E,n)$ and have, tautologically,
$\Hat\bC^\infty(T^*[n]E)= \Hat\bcalX(E,n)$
and, noncanonically,
$T^*[n]E\simeq E\oplus T^*[n]M\oplus E^*[n]$.
\fine\end{example}

Given two smooth graded manifolds $\calM$ and $\calN$, one defines their Cartesian product $\calM\times\calN$ as the
smooth graded manifold whose algebra of functions is $\bC^\infty(\calM)\Hat\otimes\bC^\infty(\calN)$
(or $\Hat\bC^\infty(\calM)\Hat\otimes\Hat\bC^\infty(\calN)$ in the hatted category).
\begin{remark}[Graded maps]\label{r:map}
Unlike in the category of manifolds, in general $\Mor(\calL\times\calM,\calN)$ is not the same as $\Mor(\calL,\Mor(\calM,\calN))$
even allowing infinite-dimensional objects. However, one can show that, given $\calM$ and $\calN$, the functor defined by
$\calL\mapsto\Mor(\calL\times\calM,\calN)$ is representable by an infinite-dimensional smooth graded manifold \cite{Var,Royun}
denoted by
$\Map(\calM,\calN)$; viz.,
$\Mor(\calL\times\calM,\calN) = \Mor(\calL,\Map(\calM,\calN))$.
Similarly, there is a hatted version denoted by $\widehat\Map(\calM,\calN))$.

For $\calN=V$ a graded vector space, one can use \eqref{e:mor}\footnote{The equation holds also for an infinite-dimensional graded vector space $V$, 
if one works from the beginning in terms of coalgebras instead of algebras of functions so as to avoid taking double duals.}
and realize the graded manifolds of maps as graded vector spaces. Namely, one can easily show that
\begin{equation}\label{e:map}
\Map(\calM,V) \simeq V\otimes\bC^\infty(\calM),\qquad
\widehat\Map(\calM,V) \simeq V\otimes\Hat\bC^\infty(\calM).
\end{equation}
In particular, one has the useful identities $\bC^\infty(\calM)\simeq\Map(\calM,\bbR)$,
$\Mor(\calM,V)=\Map(\calM,V)_0$, 
$\Map(\calM,V[k])=\Map(\calM,V)[k]$, $\Map(\calM,V\oplus W)=\Map(\calM,V)\oplus\Map(\calM,W)$,
and their hatted versions.
\fine
\end{remark}

On a graded manifold we can then define the notions of vector fields, multivector fields, Berezinian integration, divergence operator.
In particular, if $\calM$ is a smooth graded manifold with algebra of functions isomorphic to $\Hat\bC^\infty(E)$ for some
graded vector bundle $E$,
we have that $\Hat\bcalX(\calM,n):=\Hat\D(\Hat\bC^\infty(\calM),n)$ is isomorphic to $\Hat\bcalX(E,n)$, so it is a graded algebra of smooth functions.
We denote $\Spec(\Hat\bcalX(\calM,n))$ by $T^*[n]\calM$ and have, tautologically,
\begin{equation}\label{e:T*n}
\Hat\bC^\infty(T^*[n]\calM)=\Hat\bcalX(\calM,n),
\end{equation}
and, noncanonically,
\begin{equation}\label{e:T*nnoncan}
T^*[n]\calM\simeq E\oplus T^*[n]M\oplus E^*[n].
\end{equation}

\begin{remark}[Multidifferential operators]\label{r:mdogm}
Multidifferential operators  may be defined as in~\ref{s:mdo}. We will use the notations $\bcalD(\calM)$ and $\Hat\bcalD(\calM)$ for
the DGLAs $\calD(\bC^\infty(\calM))$ and $\calD(\Hat\bC^\infty(\calM))$. The HKR maps
$\bcalX(\calM)\to\bcalD(\calM)$ and $\Hat\bcalX(\calM)\to\Hat\bcalD(\calM)$
are quasiisomorphisms of differential complexes \cite{CFrel} (see also \cite{CFL}).
\fine\end{remark}


\subsubsection{Poisson structures}
A smooth graded manifold $\calM$ is called a graded Poisson manifold of degree $n$ if $\Hat\bC^\infty(\calM)$
is endowed with a bracket that makes it into an $n$\ndash Poisson algebra.
By \eqref{e:T*n},
for every smooth graded manifold $\calM$, $T^*[n]\calM$ is a Poisson manifold of degree $n$ in a canonical way.
As a Poisson bracket is  a graded biderivation, an $n$\ndash Poisson structure on
$\Hat\bC^\infty(\calM)$ determines a tensor field $\pi$ of rank two. The shifted graded skew-commutativity may be
taken into account \cite{RajTh}
by regarding $\pi$ as an $(n+1)$\ndash shifted
bivector field of degree $-n$ on $\calM$, i.e., an element
of $(S^2_{\Hat\bC^\infty(\calM)}(\Der(\Hat\bC^\infty(\calM))[-1-n]))_{2+n}$. The Jacobi identity for the Poisson
bracket is then equivalent to the equation $\Lie\pi\pi=0$. A bivector field of degree $-n$ satisfying this equation
will be called an $n$\ndash Poisson bivector field. The Poisson bracket of two functions $f$ and $g$ may then
be recovered as the derived bracket 
\begin{equation}\label{e:derPoiss}
\Poiss fg = \Lie{\Lie f\pi}g, 
\end{equation}
where $f$ and $g$ are regarded on the r.h.s.\ as $0$\ndash vector fields.

If the $n$\ndash Poisson structure of a graded Poisson manifold is nondegenerate, we speak of a graded symplectic manifold
of degree $n$.
So, $T^*[n]M$ is a graded symplectic manifold of degree $n$ in a canonical way.\footnote{It may be proved \cite{Schw}
that every graded symplectic manifold of degree $2k+1$ is isomorphic to some $T^*[2k+1]\calM$ with
canonical symplectic structure.} We call (anti)symplectomorphism between two graded symplectic manifolds
a morphism of the underlying smooth graded manifolds that yields an (anti)isomorphism of the Poisson algebras of functions.
We have the following fundamental
\begin{theorem}[Legendre mapping \cite{RoyTh}]\label{t:Leg}
Let $E$ be a graded vector bundle. Then $T^*[n]E$ is canonically antisymplectomorphic to $T^*[n](E^*[n])$
for all $n$.
\end{theorem}
Observe that \eqref{e:T*nnoncan} implies that the two graded manifolds in the Theorem are diffeomorphic. The additional statement is that
there is a diffeomorphism preserving Poisson brackets up to a sign and that it is canonical (i.e., independent of the choice of connection
used to prove \eqref{e:T*nnoncan}). For a proof, see \cite{RoyTh}.
\begin{remark}
The name ``Legendre mapping'' comes from the simplest instance \cite{Tulc} of this theorem  in the category of manifolds,
$T^*TM\simeq T^*T^*M$, which induces the usual Legendre transformation of functions.
The generalization $T^*E\simeq T^*E^*$ is due to \cite{MX}. The explicit expression in coordinates of this map
also suggests the name of ``Fourier transformation'' which is used in
\cite{CFrel}.
\fine
\end{remark}

\subsection{Further readings}
In this short introduction we did not consider: local coordinates, the definition of graded manifolds as ringed spaces, 
differential and integral forms as well as a proper definition of graded submanifolds and of infinite-dimensional graded manifolds.
We refer to \cite{Roy} and references therein for further reading on graded manifolds. For supermanifolds, see also \cite{Baech,Ber,DM,Lei,Var}.

\section{The BV formalism}\label{s:BV}
We give here a presentation of the BV formalism \cite{BV,FF} (which is a generalization of the BRST formalism \cite{BRST,Tyu}) 
based mainly on \cite{Schw}.
See also \cite{Ans1,Ans2,CatBV,Fio,GPS,HT}.

\subsection{de~Rham theory revisited}
Let $M$ be a smooth orientable manifold with a volume form $v$ and
$\phi_v$ the isomorphism defined in Example~\ref{E:T*odd}.
Define $\Delta_v:=\phi_v^{-1}\circ\dd\circ\phi_v$ where $\dd$ is the exterior derivative.
(Observe that $\Delta_v$ restricted to vector fields is just the divergence operator.)
So $\Delta_v^2=0$.
Since $\phi_v$ is not an algebra morphism, $\Delta_v$ is not a derivation;
one can however show that
\begin{equation}\label{e:DeltaXY}
\Delta_v(XY)=\Delta_v(X)Y +(-1)^i X\Delta_v(Y) + (-1)^{i+1} \Lie XY,\quad
X\in\calX^i(M),\ Y\in\calX(M).
\end{equation}
Since $\phi_v(X)$ is a differential form, it is natural to integrate it on a submanifold of the corresponding degree.
Stokes' Theorem may then be reformulated by saying that the integral vanishes if $X$ is $\Delta_v$\ndash exact,
and that it is invariant under cobordisms if $X$ is $\Delta_v$\ndash closed.
Using the language of smooth graded manifolds as in 
Example~\ref{E:T*odd}, we then have the
\begin{theorem}\label{t:Stokes}
Let $v$ be a volume form on $M$ and
$X$ a function on $T^*[k]M$, $k$ odd. Then:
\begin{enumerate}
\item $\int_{L_C} X\,\surd\rho_v = \int_{L_{C'}} X\,\surd\rho_v$ for every two cobordant submanifolds $C$ and $C'$ of $M$
if{f} $X$ is $\Delta_v$\ndash closed.
\item $\int_{L_C} X\,\surd\rho_v = 0$ for every $C$ if{f} $X$ is $\Delta_v$\ndash exact.
\end{enumerate}
\end{theorem}
Let $Q_X:=\Lie X{\ }$ denote the Hamiltonian vector field of $X\in\bC^\infty(T^*[k]M)\simeq\calX(M,k)$, $k$ odd.
Using \eqref{e:DeltaXY} and Stokes' Theorem, one easily has the following characterization of 
$\Delta_v$ in terms of the canonical symplectic structure of $T^*[k]M$:
\begin{theorem}\label{t:Deltadiv}
$\Delta_v X =\frac12 \Div_{\rho_v} Q_X$ for every volume form $v$.
\end{theorem}
By Lemma~\ref{l:murho}, we know that every Berezinian form on $T^*[k]M$ may be written, up to a constant, as $\ee^\sigma\rho_v=:\rho_v^\sigma$
for some volume form $v$ and some function $\sigma$. We write 
$\surd\rho_v^\sigma:=\ee^{\frac\sigma2}\,\surd\rho_v$. By Theorem~\ref{t:Stokes}, 
$\int_{L_C}\surd\rho_v^\sigma$ is the same for all cobordant submanifolds if{f} $\ee^{\frac\sigma2}$ is $\Delta_v$\ndash closed.
Assuming for simplicity $\sigma$ to be even, by Theorem~\ref{t:Deltadiv} and Remark~\ref{r:divfX},
one can show that this is the case if{f}
\begin{equation}\label{e:QMEsigma}
\Delta_v\sigma + \frac14\Lie\sigma\sigma = 0.
\end{equation}
Given a solution $\sigma$ of this equation, one can define a new coboundary operator $\Omega_{v,\sigma}:=\Delta_v+\frac12 Q_\sigma$. Remark that
$\Omega_{v,\sigma}X = \ee^{-\frac\sigma2}\Delta_v(\ee^{\frac\sigma2}X)$.
Thus, multiplication by $\ee^{\frac\sigma2}$ is an invertible chain map
$(\bC^\infty(T^*[k]M),\Omega_{v,\sigma})\to(\bC^\infty(T^*[k]M),\Delta_v)$
and the two cohomologies are isomorphic. Moreover, Theorem~\ref{t:Stokes} is still true if one replaces $(\rho_v,\surd\rho_v,\Delta_v)$
by $(\rho_v^\sigma,\surd\rho_v^\sigma,\Omega_{v,\sigma})$. 

{\small
\subsection{The general BV formalism}
Even though the above setting is all we need in the present paper, for completeness we give an overview of the general results of \cite{Schw}.
For this one needs the notion of submanifold of a graded manifold as well as notions of symplectic geometry on graded manifolds
which we are not going to introduce here.
\begin{theorem}
Let $k$ be an odd integer. Then:
\begin{enumerate}
\item Theorem~\ref{t:Stokes} holds if $M$ is a graded manifold and $v$ a Berezinian form.
\item Every graded symplectic manifold of degree $k$ is symplectomorphic to some $T^*[k]M$ with canonical symplectic form.
\item There is a canonical way (up to a sign) of restricting a Berezinian form $\rho_v$ on $T^*[k]M$ to a Berezinian form
denoted by $\surd\rho_v$ on a Lagrangian submanifold.
\item Every Lagrangian submanifold $L$ of $T^*[k]M$ may be deformed to a Lagrangian submanifold of the form $L_C$,
with $C$ a submanifold of $M$.
\item If $X$ is $\Delta_v$\ndash closed, then $\int_{L} X\,\surd\rho_v=\int_{L'} X\,\surd\rho_v$ if $L$ may be deformed to $L'$.
\item If $X$ is $\Delta_v$\ndash exact, then $\int_{L} X\,\surd\rho_v=0$ for every Lagrangian submanifold $L$.
\end{enumerate}
\end{theorem}

\subsubsection{Generating functions}\label{s:genfun}
To do explicit computations, it is useful to describe the Lagrangian submanifold in terms of generating functions.
Generalizing concepts from symplectic geometry to graded manifolds, one sees that the graph of the differential of a function of
degree $k$ on $M$ is a Lagrangian submanifold of $T^*[k]M$. Such a function is called a generating function.
However, Lagrangian submanifolds of this form project onto $M$; so certainly a conormal
bundle cannot be represented this way. 

A slightly more general setting is the following. We assume here some knowledge of symplectic geometry (see e.g. \cite{WB}) and generalize
a classical construction.
Let $U$ be an auxiliary graded manifold, and
let $f$ be a function of degree $k$ on $M\times U$. Let $\Sigma$ be the $U$\ndash critical set of $f$; i.e.,
the subset of $M\times U$ where the differential of $f$ along $U$ vanishes. Assume $\Sigma$ to be a submanifold and let
$\phi\colon\Sigma\to T^*M$ be defined by $(x,u)\mapsto (x,\dd f(x,u))$. Then $\phi$ is a Lagrangian immersion whose image we denote by 
$L(f)$.

For example, if $C$ is a submanifold of $M$ defined by global regular constraints $\phi_1,\dots,\phi_r$,
with $\phi_j$ of degree $n_j$,
we may take $U:=\bigoplus_{j=1}^r\bbR[n_j-k]$ and define $\Psi=\sum_j \beta^j\phi_j$, 
where $\beta^j$ is the coordinate on $\bbR[n_j-k]$.
It turns then out that $L(\Psi)=N^*[k]C$.\footnote{In the absence of global regular constraints, conormal bundles may be described
by a further generalization of generating functions, the so-called Morse families.
See, e.g., \cite{WB}.}
We regard now $\Psi$ as a function on $\Tilde M:=M\times U\times U[-k]$ and denote by $L_\Psi$ the graph of its differential.
On $U\times U[-k]$, we choose the Lebesgue measure for the even components and a pure Berezinian form for the odd ones.
We denote by $\Tilde v$ the Berezinian form on $\Tilde M$ obtained by this times $\rho_v$.
Finally, let $u$ be the pairing between $U$ and $U^*$ regarded as a function of degree zero on $U[-k]\times U^*[k]$ and hence,
by pullback, on $T^*[k]\Tilde M$.
Then a simple computation (using the Fourier representation of the delta function) shows that
\[
\int_{N^*[k]C} X \ee^{\frac\sigma2}\,\surd\rho_v = \int_{L_\Psi} X \ee^{\frac\sigma2+\ii u}\,\surd\rho_{\Tilde v}.
\]
Observe that deforming $\Psi$ just deforms the Lagrangian submanifold (which in general will no longer be a conormal bundle)
but leaves the result unchanged.
}

\subsection{BV notations}
The BV formalism consists of the above setting with $k=-1$ (for historical reasons). The $-1$\ndash Poisson bracket is called BV
bracket and usually denoted by $\BV{\ }{\ }$. The coboundary operator $\Delta_v$ is called the BV Laplacian, has degree $1$ and, as $v$ is fixed,
is usually simply denoted by $\Delta$.
A solution $\sigma$ to \eqref{e:QMEsigma} is usually written as $\sigma=2\frac\ii\hbar \sfS$, where $\hbar$ is a parameter and
$\sfS$, called the BV action,  
is assumed to be of degree $0$ (so that $Q_\sfS$ is of degree $1$) and is allowed to depend on $\hbar$. It satisfies 
the so-called ``quantum master equation'' (QME)
$\BV \sfS\sfS -2\ii\hbar\Delta \sfS=0$.
The coboundary operator $\Omega_{v,\sigma}$ is then also homogeneous of degree $1$. Setting $\Omega:=-\ii\hbar\Omega_{v,\sigma}$,
we have $\Omega=Q_\sfS-\ii\hbar\Delta$.
An $\Omega$\ndash closed element $\scalO$ is called an observable, and its expectation value
\begin{equation}\label{e:BVvev}
\vev\scalO := \frac{\int_L \ee^{\frac\ii\hbar \sfS}\, \scalO\,\surd\rho_v}{\int_L \ee^{\frac\ii\hbar \sfS}\,\surd\rho_v}
\end{equation}
is  invariant under deformations of $L$. 
The choice of an $L$ goes under the name of 
gauge fixing.\footnote{\label{r:genfun}This is usually done as explained in~\ref{s:genfun} by using an auxiliary space and a generating function $\Psi$
which is in this case of degree $-1$ and is called the gauge-fixing fermion.}
\begin{remark}[Ward identities]\label{r:Ward}
Expectation values of $\Omega$\ndash exact observables vanish, but they may lead to interesting
relations called Ward identities.
\fine
\end{remark}
\begin{remark}\label{r:Shbar}
One often assumes $\hbar$ to be ``small.'' Actually, one even takes $\sfS$ to be a formal power series in $\hbar$,
$\sfS=\sum_{i=0}^\infty \hbar^i \sfS_i$. Then $\sfS_0$ satisfies the ``classical master equation'' (CME) $\BV \sfS\sfS=0$ and
$Q_{\sfS_0}$ is a coboundary operator (sometimes called the BRST operator). One may look for solutions of the QME
starting from a solution $\sfS_0$ of the CME\@. One easily sees that there is a potential obstruction to doing this
(the so-called anomaly)
in the second cohomology group of $Q_{\sfS_0}$.
\fine
\end{remark}
\begin{remark}\label{r:defBV}
An observable $\scalO$ of degree zero may also be thought of as an infinitesimal deformation of the BV action, for $\sfS+\epsilon \scalO$ then
satisfies the CME up to $\epsilon^2$. For this to be a finite deformation,
we should also assume $\BV\scalO\scalO=0$. 
\fine
\end{remark}

\subsection{Applications}
Suppose that the integral of $\ee^{\frac\ii\hbar \sfS}$ along a Lagrangian submanifold $L$ 
is not defined, but that it is enough to deform
$L$ a little bit for the integral to exist. Then one defines the integral along $L$ as the integral along a deformed Lagrangian submanifold $L'$. 
For a given cobordism class of deformations, the integral does not depend on the specific choice of $L'$
if $S$ is assumed to satisfy the QME\@. This is really analogous to the definition of the principal part of an integral \cite{Fio}.

The typical situation is the following: One starts with a function $S$
defined on some manifold $\calM$. One assumes there is a (nonnecessarily integrable) distribution
on $\calM$---the ``symmetries''---under 
which $S$ is invariant. One then adds odd variables of degree $1$ (the generators of the distribution, a.k.a.\ the ghosts) defining a graded manifold $\Tilde\calM$ which
fibers over $\calM$ and is endowed with a vector field $\delta$ that describes the distribution. Then one tries to extend
$S$ to a solution $\sfS_0\in\bC^\infty(T^*[-1]\Tilde\calM)$  of the CME  such that $Q_{\sfS_0}$
and $\delta$ are related vector fields. Under the assumption that the original distribution is integrable on the subset (usually assumed to be a submanifold) 
of critical
points of $S$, one can show that this is possible under some mild regularity assumptions \cite{BV}. The next step is to find a solution of the QME
as in Remark~\ref{r:Shbar} if there is no anomaly.

Because of the invariance of $S$, the integral of $\ee^{\frac\ii\hbar S}$ on $\calM$ will diverge (if the symmetry directions
are not compact). On the other hand, if we integrate over $\Tilde\calM$, we also have zeros corresponding to the odd directions which we have introduced and along
which $S$ is constant. If we introduce all generators of symmetries, we have as many zeros as infinities, so there is some hope to make this ill-defined
integral finite. This is actually what happens if we find a solution of the QME as in the previous paragraph and integrate on
a different Lagrangian submanifold of $T^*[-1]\Tilde\calM$ than its zero section $\Tilde\calM$.

Given a function $\calO$ on $\calM$, it makes sense to define its expectation value as in \eqref{e:BVvev}
if there is an observable $\scalO$ whose restriction to $\calM$ is $\calO$.

\begin{remark}[Field theory]\label{r:ft}
In field theory one considers integrals  of the form \eqref{e:vev} with $\calM$ infinite dimensional.
Integration around critical points is defined by expanding the nonquadratic part of $S$ and evaluating Gaussian expectation values.
If there are symmetries, the critical points are degenerate and one cannot invert the quadratic form. One then operates as above getting
an 
integral with the quadratic part of the BV action  nondegenerate, so one can start the perturbative 
expansion.\footnote{In order to have Gaussian integration on a vector space, one defines integration along the chosen Lagrangian submanifold
via a generating function as explained in~\ref{s:genfun} and in footnote~\ref{r:genfun}.}
This is not the end of the story since two problems arise. The first is that the formal evaluation of the Gaussian expectation values
leads to multiplying distributions. The consistent procedure for overcoming this problem, when possible, goes under the name of renormalization.
The second problem is that, in the absence of a true measure, there is no divergence operator and thus no well-defined BV Laplacian $\Delta$.
This is overcome by defining $\Delta$ appropriately in perturbation theory. On the other hand, the BV bracket is well-defined
(on a large enough class of functions).
In the present paper the field theory is so simple that renormalization is (almost) not needed, so we will not talk about it.
On the other hand, 
it makes sense \cite{CF}
to assume that $\Delta$ exists and vanishes on the local functionals we are going to consider, while on products thereof one uses \eqref{e:DeltaXY}.
\fine
\end{remark}

\section{BV 2D TFT}
We go back now to our original problem described in the Introduction. This may also be regarded as a continuation of our presentation in
\cite[Part III]{BCKT}.
\subsection{The BV action}
We start by considering the TFT with action \eqref{e:S1} and symmetries \eqref{e:delta1}.
We promote the generators $\beta$ of the symmetries to odd variables of degree $1$; i.e., we define $\Tilde\calM_1=\calM_1\oplus\Omega^0(\Sigma)[1]$ and
the vector field $\delta$ by its action on the linear functions $\xi$, $\eta$ and $\beta$:
$\delta\xi=0$, $\delta\eta=\dd\beta$, $\delta\beta=0$.
Using integration on $\Sigma$, we identify $T^*[-1]\Tilde\calM_1$ with 
$\Tilde\calM_1\oplus\Omega^2(\Sigma)[-1]\oplus\Omega^1(\Sigma)[-1]\oplus\Omega^2(\Sigma)[-2]$ and denote
the new coordinates, in the order, by $\xi^+$, $\eta^+$ and $\beta^+$. We introduce the ``superfields''
$\sfxi=\xi+\eta^++\beta^+,$ $\sfeta=\beta+\eta+\xi^+$,
and define
\begin{equation}\label{e:BVS}
\sfS(\sfxi,\sfeta):=\int_\Sigma\sfeta\,\dd\sfxi,
\end{equation}
where by definition the integration selects the $2$\ndash form.
It is not difficult to see that $\sfS$ satisfies the CME and $\sfS_{|_{\calM_1}}=S$.
Moreover, the action of $Q_\sfS$ on the coordinate functions may be summarized in
\begin{equation}\label{e:BVS1}
Q_\sfS\sfxi=\dd\sfxi,\qquad Q_\sfS\sfeta=\dd\sfeta.
\end{equation}
So $Q_\sfS$ and $\delta$ are related vector fields.

By \eqref{e:OmegaM}, we may regard $\sfxi$ as an element of $\Omega(\Sigma)$ and $\sfeta$ as an element of $\Omega(\Sigma)[1]$.
As $\Omega(\Sigma)=\bC^\infty(T[1]\Sigma)$,
by Remark~\ref{r:map} at the end we may further identify $\Omega(\Sigma)$ with $\Map(T[1]\Sigma,\bbR)$ and
$\Omega(\Sigma)[1]$ with $\Map(T[1]\Sigma,\bbR[1])$ or, equivalently, with $\Map(T[1]\Sigma,\bbR^*[1])$.
The latter choice is more appropriate in view of \eqref{e:BVS} where we pair $\sfxi$ with $\sfeta$.
By Remark~\ref{r:map} at the end again, we have eventually the identification
$T^*[-1]\Tilde\calM_1\simeq\Map(T[1]\Sigma,T^*[1]\bbR)$,
where we have identified $\bbR\oplus\bbR^*[1]$ with $T^*[1]\bbR$ (by the results of Example~\ref{ex:Tstar} with $E=\bbR$ as a vector
bundle over a point). This is actually the viewpoint taken in \cite{AKSZ} (see also \cite{CFAKSZ}).
Finally, observe that we may also regard $T^*[-1]\Tilde\calM_1$ as $\Map(T[1]\Sigma,T^*[1]\bbR[0])$ 
if we wish to consider formal power series in the coordinate functions.

The ill-defined integration on $\tilde\calM_1$ is now replaced by a well-defined (in the sense of perturbation theory) 
integration over another Lagrangian submanifold $L$ of $T^*[-1]\Tilde\calM_1$. For example, as in \cite{CF}, we may take
$L=N^*[-1]C$ where $C$ is the submanifold 
of $\Tilde\calM_1$ defined as the zero locus of $\dd*\eta$, where the Hodge-star operator
is defined upon choosing a volume form on $\Sigma$.

\subsection{The superpropagator}
The main object appearing in the explicit evaluation of expectation values of functions of $\sfxi$ and $\sfeta$ is the
``superpropagator'' $\vev{\sfxi(z)\sfeta(w)}$, where $z$ and $w$ are points in $\Sigma$.
Independently of the choice of gauge fixing, we have the Ward identity
\begin{multline*}
0=\vev{\Omega(\sfxi(z)\sfeta(w))}=
\vev{Q_\sfS(\sfxi(z)\sfeta(w))}-\ii\hbar\vev{\Delta(\sfxi(z)\sfeta(w))}=\\
=\dd\vev{(\sfxi(z))\sfeta(w)}-\ii\hbar\vev{\BV{\sfxi(z)}{\sfeta(w)}}=
\dd\vev{(\sfxi(z))\sfeta(w)}-\ii\hbar\delta(z,w),
\end{multline*}
where we assumed $\Delta(\sfxi(z))=\Delta(\sfeta(w))=0$ (which is consistent with perturbation theory) and $\delta$ denotes
the delta distribution (regarded here as a distributional $2$\ndash form).
Thus, we get the fundamental identity\footnote{This method for deriving properties of the superpropagator just in terms
of Ward identities works also for the higher-dimensional generalization of this TFT \cite{CR}.}
\begin{equation}\label{e:dvev}
\dd\vev{(\sfxi(z)\sfeta(w)}=\ii\hbar\delta(z,w).
\end{equation}
The restriction of the superpropagator to the configuration space $C_2(\Sigma):=\{(z,w)\in\Sigma\times\Sigma:z\not=w\}$
is then a closed, smooth $1$-form. Namely, if we set
$\ii\hbar\theta(z,w):=\vev{\sfxi(z)\sfeta(w)}$, $(z,w)\in C_2(\Sigma)$,
then $\theta\in\Omega^1(C_2(\Sigma))$ and  $\dd\theta=0$. We call it the propagator $1$\ndash form.
The delta distribution in \eqref{e:dvev}
implies that $\int_\gamma\theta=1$
where $\gamma$ is generator of the singular homology of $C_2(\Sigma)$
(viz., $\gamma$ is a loop of $w$ around $z$). Observe that $\theta$ is defined up to an exact $1$\ndash form.
Different choices of gauge fixing just correspond to different, but cohomologous, choices of $\theta$.

If $\de\Sigma\not=\emptyset$, we have to choose boundary conditions. Repeating the considerations in the Introduction, we see that there are
two possible boundary conditions compatible with \eqref{e:BVS1}; viz.:
\begin{align*}
 &\text{Neumann boundary conditions:} & \iota_{\de\Sigma}^*\sfeta &=0, & (N)\\
 &\text{Dirichlet boundary conditions:}  & \iota_{\de\Sigma}^*\sfxi &=0, & (D)
\end{align*}

For $\de\Sigma=\emptyset$, the BV action \eqref{e:BVS} is invariant under the exchange of $\sfeta$ with $\sfxi$.
This implies that $\psi^*\theta=\theta$ with $\psi(z,w)=(w,z)$.\footnote{The cohomology class of a propagator
$1$\ndash form is necessarily $\psi$\ndash invariant. The stronger condition is that it is
$\psi$\ndash invariant without passing to cohomology.}
For $\de\Sigma\not=\emptyset$,
we denote by $\theta_N$ and $\theta_D$ the propagator $1$\ndash forms corresponding to N and D boundary conditions, respectively.
These $1$\ndash forms have to satisfy in addition boundary conditions.
Let $\de_i C_2(\Sigma)=\{(z_1,z_2)\in C_2(\Sigma):z_i\in\de\Sigma\}$ and $\iota_i$ the inclusion of $\de_i C_2(\Sigma)$ into
$C_2(\Sigma)$. Then we have
$\iota_1^*\theta_D=0$ and $\iota_2^*\theta_N=0$.
These $1$\ndash forms are no longer invariant under the involution $\psi$ defined above; they are instead related by it:
viz., $\psi^*\theta_N=\theta_D$.

\subsection{Duality}\label{s:dual}
Exchanging the superfields has a deeper meaning. Observe that the $0$\ndash form component $\xi$
of $\sfxi$ is an ordinary function (of degree zero), while the $0$\ndash component form $\beta$ of $\sfeta$ has been assigned degree $1$
and has values in $\bbR^*$.
So, when we make this exchange, we are actually trading, loosely speaking, a map $\xi\colon\Sigma\to\bbR[0]$ for a map $\beta\colon\Sigma\to\bbR^*[1]$.
In exchanging the superfields, we are then actually performing the canonical symplectomorphism $\Map(T[1]\Sigma,T^*[1]\bbR[0])\to\Map(T[1]\Sigma,T^*[1]\bbR^*[1])$
which is induced by the canonical symplectomorphism $T^*[1]\bbR[0]\to T^*[1]\bbR[1]$, a special case of the Legendre mapping of Theorem~\ref{t:Leg}.
If we now take the graded  vector space $\bbR[k]$ as target,
the superfield exchange is a symplectomorphism $\Map(T[1]\Sigma,T^*[1]\bbR[k])\to\Map(T[1]\Sigma,T^*[1]\bbR^*[1-k])$.
In conclusion, the TFT with target $\bbR[k]$ is equivalent to the TFT with target $\bbR^*[1-k]$ if $\Sigma$ has no boundary; whereas, if $\Sigma$ has a boundary,
the TFT with target $\bbR[k]$ and N boundary conditions is equivalent to the TFT with target $\bbR^*[1-k]$ and D boundary conditions.
Thus, upon choosing the target appropriately, one may always assume to have only N boundary conditions.

\subsection{Higher-dimensional targets}\label{s:hdt}
We may allow a higher-dimensional target as in \eqref{e:SRn} or in \eqref{e:SM}, and  it makes sense for it to be a graded vector space
or a graded manifold $M$. 
Now the space of fields may be identified with
$\Map(T[1]\Sigma,T^*[1]M)$. For simplicity, assume the target to be a graded vector space $V$ (which is anyway the local version of the general case).
Upon choosing
a graded basis $\{e_I\}$ and its dual basis $\{e^I\}$, we may consider the components $\sfxi^I$ and $\sfeta_I$ of the superfields. The superpropagator
may then be computed as
$\vev{\sfxi^I(z)\sfeta_J(w)}=\ii\hbar\theta(z,w)\delta^I_J$, $(z,w)\in C_2(\Sigma)$,
where $\theta$ is the $1$\ndash form propagator of the TFT with target $\bbR$. Again we are allowed to exchange superfields, but we may 
decide to exchange only
some of them. Let $V=W_1\oplus W_2$. A superfield exchange corresponding to $W_2$\ndash components establishes a symplectomorphism
$\Map(T[1]\Sigma,T^*[1](W_1\oplus W_2)\simeq\Map(T[1]\Sigma,T^*[1](W_1\oplus W_2^*[1])$. If we have N boundary conditions
on the $W_1$\ndash components and D boundary conditions on the $W_2$\ndash components, 
the exchange yields a theory with only N boundary conditions.

If we work with target a graded manifold $M$ and D boundary conditions on a graded submanifold $C$, the perturbative expansion actually
sees as target the graded submanifold $N[0]C$ of $M$ (as a formal neighborhood of $C$). As a consequence of the previous considerations, this
is the same as the TFT with target $N^*[1]C$ and N boundary conditions. This case has been studied in \cite{CFb,CFrel}.


\subsubsection{Assumptions}
{}From now on we assume that $\Sigma$ is the disk and that on its boundary $S^1$ we put N boundary conditions.
We also choose a point $\infty\in S^1$ and fix the map $X$ to take the value $x\in M$ at $\infty$.
By setting $X=x+\xi$ we identify the theory with target $M$ with the theory with target the graded vector space $T_xM[0]$.
The superfield $\sfxi\in\Map(T[1]\Sigma,T_xM[0])$ is then assumed to vanish at $\infty$.

\subsection{Ward identities and formality theorem}\label{s:Ward}
There exists a class of interesting observables associated to multivector fields on the target. For simplicity we assume the target to be
a graded vector space $V$, make the identification \eqref{e:calXV} and use a graded basis. So,
for a $k$\ndash vector field $F\in\bcalX(V)$, we define
\begin{equation}\label{e:bulk}
\sfS_F(\sfxi,\sfeta) = \frac1{k!}\int_\Sigma F^{i_1\dots i_k}(\sfxi)\sfeta_{i_1}\cdots\sfeta_{i_k}.
\end{equation}
Since $Q_\sfS\sfS_F= \frac1{k!}\int_{\de\Sigma} F^{i_1\dots i_k}(\sfxi)\sfeta_{i_1}\cdots\sfeta_{i_k}$., we have defined an observable
unless $F$ is a $0$\ndash vector field (i.e., a function), 
for one may show \cite{CF} that it is consistent
to assume $\Delta\sfS_F=0$. 
We will call observables of this kind bulk observables. By linear extension, we may associate a bulk  observable to every element $F\in\Hat\bcalX(V)$.
If $F$ is of total degree $f$, then $\sfS_F$ is of degree $f-2$.
One may also show \cite{CF} (see also \cite{CFAKSZ}) that $\BV{\sfS_F}{\sfS_G}=\sfS_{\Lie FG}$ for any two multivector fields $F$ and $G$.
Another interesting class of observables is associated to functions on the target. Given a function $f$ and a point $u\in\de\Sigma$, we set
$\scalO_{f,u}(\sfxi,\sfeta) = f(\sfxi(u)) = f(\xi(u))$.
Since $Q_\sfS\scalO_{f,u}=0$ as $u$ is on the boundary, since the
difference $\scalO_{f,u}-\scalO_{f,u'}$ is equal $Q_\sfS\int_u'^u f(\sfxi)$ and since
one may consistently set to zero $\Delta$ applied to functions of $\sfxi$ only,  we have defined new observables, which we will call boundary
observables, in which the choice of $u$ is immaterial.

A product of observables is in general not an observable (since $\Omega$ is not a derivation). A product which is however an 
observable is $\scalO(F;f_1,\dots,f_k)_{u_1,\dots u_k}:=\sfS_F\scalO_{f_1,u_1}\cdots\scalO_{f_k,u_k}$, 
where $F$ is a $k$\ndash vector field, $k>0$, the $f_i$s are functions and
the $u_i$s are ordered points on the boundary. The expectation value may easily be computed \cite{CF} and one gets
$\vev{\scalO(F;f_1,\dots,f_k)_{u_1,\dots u_k}} = \HKR(F)(f_1\otimes\dots\otimes f_k)$.
More generally, one may define
\[
\scalO(F_1,\dots,F_m;f_1,\dots,f_k)_{u_1,\dots u_k}:=\sfS_{F_1}\cdots\sfS_{F_m}\scalO_{f_1,u_1}\cdots\scalO_{f_k,u_k}.
\]
One may show \cite{CF}
that the expectation value of $\scalO(F_1,\dots,F_m;f_1,\dots,f_k)_{u_1,\dots u_k}$ may be regarded as
a multidifferential operator $U_m(F_1,\dots,F_m)$ acting on $f_1\otimes\dots\otimes f_k$. This way one defines multilinear maps $U_m$s from $\bcalX$ to $\bcalD$.
However, the explicit form
of the multidifferential operators will depend on the chosen gauge fixing as $\scalO(F_1,\dots,F_m;f_1,\dots,f_k)_{u_1,\dots u_k}$ is
not an observable in general. 
One may get very interesting identities relating the $U_m$s by considering the Ward identities
\begin{equation}\label{e:Ward}
0=\vev{\Omega\scalO(F_1,\dots,F_m;f_1,\dots,f_k)_{u_1,\dots u_k}}.
\end{equation}
One may show \cite{CF,Kon} that the various contribution of the r.h.s.\ correspond to collapsing in all possible ways
some of the bulk observables together with some of the boundary observables (with consecutive $u$s). As a result one gets
relations among the $U_m$s. 
To interpret them, we have to introduce some further concepts.
\begin{definition}
An $L_\infty$\ndash algebra\footnote{We follow here the sign conventions of \cite{Vor}.} \cite{LS,S2}
is a graded vector space $V$ endowed with operations (called multibrackets) $L_k\in\Hom_1(S^kV,V)$, $k\in\bbN$,
satisfying for all $n\ge0$ and for all $v_1,\dots,v_n\in V$
\[
\sum_{k+l=n}\ \sum_{\sigma\in (k,l)\text{-shuffles}} \sign(\sigma) L_{l+1}(L_k(v_{\sigma(1)},\dots,v_{\sigma(k)}),v_{\sigma(k+1)},\dots,v_{\sigma(n)})=0,
\]
where a $(k,l)$\ndash shuffle is a permutation on $k+l$ elements such that $\sigma(1)<\dots<\sigma(k)$ and
$\sigma(k+1)<\dots<\sigma(k+l)$, while
the sign of the permutation $\sigma$
is defined by $v_{\sigma(1)}\cdots v_{\sigma(n)}=\sign(\sigma)v_1\cdots v_{n}$ in $S^kV$. We call flat an $L_\infty$\ndash algebra with $L_0=0$.
\end{definition}
In a flat $L_\infty$\ndash algebra, $L_1$ is a coboundary operator. We denote by $H(V)$ the $L_1$\ndash cohomology.
Observe that $H(V)[-1]$ acquires a GLA structure.

For $V$ finite dimensional, we may identify $\Hom_1(SV,V)$ with $(SV^*\otimes V)_1$ and so with $\bfrX(V)_1$.
An $L_\infty$\ndash algebra on $V$ is then the same as the data of a ``cohomological vector field'' (i.e., a vector field of degree $1$
that squares to zero). The same holds in the infinite-dimensional case if one defines
things appropriately.
\begin{example}
A (D)GLA $\frg$ may be regarded as a flat $L_\infty$\ndash algebra by setting $V=\frg[1]$ 
and defining $L_k$ to be the Lie bracket for $k=2$ (and the differential for $k=1$), while all other $L_k$s are set to zero.
\fine\end{example}
One may introduce the category of $L_\infty$\ndash algebras 
by defining an
$L_\infty$\ndash morphism from 
$V$ to $W$ to be a sequence
of morphisms $SV\to W$ with appropriate relations between the two sets of multibrackets.
We do not spell out these relations here. They essentially state that there is a morphism $V\to W$
as (possibly infinite-dimensional) graded manifolds such that the corresponding homological vector fields are related.
We write $U\colon V\leadsto W$ for an $L_\infty$\ndash morphisms
with components $U_m\in\Hom_0(S^mV,W)$. An important property of the definition is the following:
If $V$ and $W$ are flat and $U_0=0$, then $U_1$ is a chain map. If $U_1$ induces an isomorphism in cohomology,
one says that $U$ is an $L_\infty$\ndash quasiisomorphism. If in addition $V$ has zero differential, 
$V[-1]$ is isomorphic as a GLA to $H(W)[-1]$, and one says that $W[-1]$ is formal.
Finally we may interpret the Ward identities \eqref{e:Ward} in terms of the DGLAs
$\Hat\bcalV(M):=\Hat\bcalX(M)[1]$ and $\Hat\bcalD(M)$ as flat $L_\infty$\ndash algebras:
\begin{theorem}[Formality Theorem]
There is an $L_\infty$\ndash morphism 
$U\colon\Hat\bcalV(M)\leadsto\Hat\bcalD(M)$, with  $U_1$ the HKR map. So $U$ is an $L_\infty$\ndash quasiisomorphism
and the DGLA $\Hat\bcalD(M)$ is formal.
\end{theorem}
The Ward identities are not a full proof of the Theorem as all arguments using infinite-dimensional integrals have to be taken with
care (e.g., we have always assumed that we can work with the BV Laplacian $\Delta$ which is actually not properly defined).
They however strongly suggest that such a statement is true. One may check that this is the case by inspecting
the finite-dimensional integrals (associated to the Feynman diagrams) appearing in the perturbative expansion.
For $M$ an ordinary smooth manifold, the Formality Theorem has been proved by Kontsevich in \cite{Kon}.
For a proof when $M$ is a smooth graded manifold, see \cite{CFrel}.

\subsection{Deforming the action: The Poisson sigma model}
As we observed in Remark~\ref{r:defBV}, an observable of degree zero that commutes with itself may be used to deform the BV action.
By considering bulk observables \eqref{e:bulk}, we get a deformed BV action $\sfS_F^\mathrm{def}=\sfS+\epsilon\sfS_F$ for every
$F=\sum_i F_i\in\bcalX(M)_2$, with $F_i$ an $i$\ndash vector field, $F_0=0$, and $\Lie FF=0$.

An element $x$ of degree one of a DGLA is called an MC (for Maurer--Cartan) element if $\dd x+\frac12\Lie xx=0$.
So $F$ must be in particular an MC element in $\Hat\bcalV(M)$. A multivector field $F$ is completely characterized by its derived brackets
\begin{gather*}
\lambda_i(a_1,\dots,a_i):=
\mathit{pr}\left(\Lie{\Lie{\cdots\Lie{\Lie F{a_1}}{a_2}}\dots}{a_i}\right)=\\
=\Lie{\Lie{\cdots\Lie{\Lie {F_i}{a_1}}{a_2}}\dots}{a_i},
\qquad a_1,\dots,a_i\in \Hat\bC^\infty(M),
\end{gather*}
where $\mathit{pr}$ is the projection from $\Hat\bcalV(M)$ onto the abelian
Lie subalgebra $\Hat\bC^\infty(M)$. A consequence of a more general results in \cite{Vor} is that
$F$ is MC if{f} $(\Hat\bC^\infty(M),\lambda)$ is an $L_\infty$\ndash algebra. The condition $F_0=0$ is precisely the condition
that this $L_\infty$\ndash algebra is flat. By construction the multibrackets $\lambda$ are multiderivations,
so we call this $L_\infty$\ndash algebra a $P_\infty$\ndash algebra ($P$ for Poisson) \cite{CFrel}.

A particular case is when $F$ is a Poisson bivector field of degree zero. This is the only possibility if the target is an ordinary manifold.
The only derived bracket is the Poisson bracket \eqref{e:derPoiss}, and 
$\sfS_F^\mathrm{def}$ is the BV action of the so-called Poisson sigma model \cite{Ike,SS}.
Another particular case is when we start with an ordinary Poisson manifold $(P,\pi)$ and consider the Poisson sigma model
with D boundary conditions on a submanifold $C$. As discussed at the end of~\ref{s:hdt}, this is the same as working with target $N^*[1]C$
and N boundary conditions. The Poisson bivector field $\pi$ induces, noncanonically, a Poisson bivector field $\Tilde\pi$ on $N[0]C$ which in turns
by the Legendre mapping yields an MC element $F$ in $\Hat\bcalV(N^*[1]C)$. As pointed out above, we need $F_0=0$. 
This is the case if{f} $C$ is a coisotropic submanifold \cite{CFb}, i.e., 
a submanifold  
whose vanishing ideal $I$  is a Lie subalgebra of 
$(C^\infty(P),\Poiss{\ }{\ })$.\footnote{According to Dirac's terminology, $C$ is determined (locally) by first-class constraints.}
The derived brackets on $\Hat\bC^\infty(N^*[1]C)$ yield the $L_\infty$\ndash algebra studied in \cite{OP}.
The zeroth $F_1$\ndash cohomology group is the Poisson algebra $C^\infty(C)^I$
of $\Poiss I{\ }$\ndash invariant functions on $C$. Hamiltonian vector fields of functions in $I$ define
an integrable distribution on $C$. The leaf space $\underline C$ is called the reduction of $C$. If it is a manifold, 
$C^\infty(\underline C)=C^\infty(C)^I$.\footnote{
We discuss here deformations of the TFT $\sfS$, i.e, the Poisson sigma model with zero Poisson structure.
If one drops the condition that the Poisson sigma model with D boundary conditions must be such a deformation, much more general
submanifolds $C$ are allowed \cite{CF1,CF2}.}

The expectation value of boundary observables in the deformed theory $\sfS_F^\mathrm{def}$ may easily be computed in perturbation theory by
expanding $\exp(\epsilon\sfS_F)$. As a result one has just to apply to the functions placed on the boundary
the formal power series of multidifferential operator
$U(\epsilon F) := \sum_{k=1}^\infty \frac{\epsilon^k}{k!}\,U_k(F,\dots,F)$.

If $\frg$ is a DGLA, by linearity one may extend the differential and 
the bracket to formal power series and so give $\epsilon\frg[[\epsilon]]$
the structure of a DGLA\@. Moreover, if $x$ is an MC element in a GLA $\frg$, then $\epsilon x$ is an MC element in $\epsilon\frg[[\epsilon]]$.
An $L_\infty$\ndash morphism $U\colon\frg\leadsto\frh$ between DGLAs $\frg$ and $\frh$ may be extended by linearity to formal power series as well.
If $X$ is an MC element in $\epsilon\frg[[\epsilon]]$, then $U(X)$ is well defined in $\epsilon\frh[[\epsilon]]$ and it may be proved to be
an MC element.

So $U(\epsilon F)$ is an MC element in $\epsilon\Hat\bcalD(M)[[\epsilon]]$. As shown in \cite{CFrel} such an MC element induces an $A_\infty$\ndash
structure on $\Hat\bC^\infty(M)[[\epsilon]]$. This is the data of multibrackets $A_i$ (with $i$ arguments)
satisfying relations analogous to those of an $L_\infty$\ndash algebra
but without symmetry requirements \cite{S1,S2}. 
If $A_0=0$, the
$A_\infty$\ndash algebra is called flat, $A_1$ is a differential for $A_2$, and the $A_1$\ndash cohomology has the structure of
an associative algebra. However, $A_0=0$ is not implied by $F_0=0$. In \cite{CFrel} it is proved that 
a potential obstruction to making  the $A_\infty$\ndash structure 
flat is contained in
the second $F_1$\ndash cohomology group. We call this potential obstruction the anomaly.


\section{Applications}
When the target $M$ is an ordinary manifold and $F$ is a Poisson bivector field,
$C^\infty(M)$ is concentrated in degree zero, so the $A_\infty$\ndash structure consists just
of the bidifferential operator and is a genuine associative algebra structure. This is the original result by
Kontsevich \cite{Kon} that every Poisson bivector field defines a deformation quantization \cite{BFFLS} of the algebra of functions.

A general method for studying certain submanifolds of so-called weak Poisson manifolds and  their quantization
has been suggested in \cite{LSh}: one concocts 
a smooth graded manifold $M$ endowed with an MC element $F$, with $F_0=0$, to describe the problem, and then applies the $L_\infty$\ndash quasiisomorphism $U$.

A particular case is the graded manifold $N^*[1]C$ associated to a coisotropic submanifold $C$, as described above.
In the absence of anomaly, 
the method yields a deformation quantization of a Poisson subalgebra of $C^\infty(C)^I$
(or of the whole algebra if the first $F_1$\ndash cohomology vanishes)  \cite{CFb,CFrel}.

A second interesting case is that of a Poisson submanifold $P'$ of a Poisson manifold $P$. The inclusion map $\iota$ is then a Poisson map (i.e., 
$\iota^*$ is a morphism of Poisson algebras). One may then try to get deformation quantizations of $P$ and $P'$ together with a morphism of
associative algebras that deforms $\iota^*$. The simplest case is when $P'$ is determined by regular constraints $\phi^1,\dots,\phi^k$.
The Koszul resolution of $C^\infty(P')$ is obtained by introducing variables $\mu^1,\dots,\mu^k$ of degree $-1$ and defining
a differential $\delta\mu^i=\phi^i$. We may interpret this differential as a cohomological vector field $Q$  on the graded manifold
$M:=P\times \bbR^k[-1]$. The Poisson bivector field $\pi$ on $P$ may also be regarded as a Poisson bivector field on $M$.
We may put the two together defining $F=Q+\pi$, which is an MC element if{f} $\Lie\pi Q=0$, i.e., if{f} the $\phi^i$s are central.
In this case $U(\epsilon F)$ produces an $A_\infty$\ndash algebra structure on $\bC^\infty(M)[[\epsilon]]$, which is flat since
$\bC^\infty(M)$ is concentrated in nonpositive degrees. 
Moreover, 
$\bC^\infty(M)_0[[\epsilon]]=C^\infty(P)[[\epsilon]]$ inherits an algebra structure which turns out to give a deformation quantization of $P$.
One may also verify that the zeroth $A_1$\ndash cohomology group $H^0$ is a deformation quantization of $P'$ and that the projection
$\bC^\infty(M)_0[[\epsilon]]\to H^0$, which is by construction an algebra morphism, is a deformation of $\iota^*$. 
By inspection of the explicit formulae, 
one may easily see that this construction is the same as the one proposed in \cite{CRy}, thus proving their
conjecture. The more general case when the regular constraints $\phi^i$ are not central, may in principle be treated following
\cite{LS} which shows the existence an MC element of the form $F=Q+\pi+O(\mu)$. Repeating the above reasoning 
does not solve the problem since in general the algebra $\bC^\infty(M)_0[[\epsilon]]$ is not associative. For this to be the case,
one has to find corrections to $F$ such that in each term the polynomial degree in the $\mu^i$s is less or equal than the polynomial degree
in the $\de/\de\mu^i$s.

A third interesting case is that of a Poisson map $J$ from a Poisson manifold $P$ to the dual of a Lie algebra
$\frg$. Under certain regularity assumptions, $J^{-1}(0)$ is a coisotropic submanifold and may be quantized as described above.
In practice, the formulae are not very explicit, even if $P$ is a domain in $\bbR^n$, for one has to choose adapted coordinates. 
A different approach is the following: First endow $P\times\frg^*$ with the unique Poisson structure which makes the projection $p_1$ to $P$
Poisson, the projection $p_2$ to $\frg^*$ anti-Poisson and such that $\Poiss{p_2^*X}{p_1^*f}_{P\times\frg^*}=p_1^*\Poiss{J_X}f_P$,
$\forall f\in C^\infty(P)$ and $\forall X\in\frg$. The graph $G$ of $J$ is then a Poisson submanifold of $P\times\frg^*$,
while $P\times\{0\}$ is coisotropic. Their intersection, diffeomorphic to $J^{-1}(0)$, 
turns out to be coisotropic in $G$. One then describes $G$ as the zero set of the regular constraints
$\phi\colon P\times\frg^*\to\frg^*$,
$(x,\alpha)\mapsto J(x)-\alpha$.
Thus, applying the above construction, one describes $G$ by an appropriate MC element $F$ on 
$M:=P\times\frg^*\times \frg^*[-1]$ and realizes the quantization of $\underline{J^{-1}(0)}$ by
the TFT with BV action $\sfS_F^\mathrm{tot}$ and D boundary conditions on $C:=P\times\{0\}\times\frg^*[-1]$.
Since we may identify $N^*[1]C$ with $\Tilde M:=P\times\frg[1]\times \frg^*[-1]$, we eventually have the TFT with target $\Tilde M$
and BV action $\sfS_{\Tilde F}^\mathrm{tot}$, where $\Tilde F$ is the Legendre transform of $F$. If $P$ is a domain in $\bbR^n$,
we may now use one coordinate chart and get explicit formulae. This construction turns out to be equivalent to the BRST method. It has
a generalization, equivalent to the BV method,
when we have a map $J\colon P\to\bbR^k$ such that $J^{-1}(0)$ is coisotropic.

All the above ideas may in principle be applied to the case when the Poisson manifold $P$ is an infinite-dimensional space of maps (or sections)
as in field theory. An $(n+1)$\ndash dimensional field theory on $M\times\bbR$ is a dynamical system
on a symplectic manifold $\calM$ of sections on $M$ (or a coisotropic submanifold thereof in gauge theories).
The Poisson sigma model version then yields \cite{Sig} an equivalent $(n+2)$\ndash dimensional field theory on $M\times\Sigma$, with $\Sigma$ the upper half plane.






\end{document}